\documentclass[12pt]{article}
\usepackage{amssymb}
\usepackage{amsfonts}
\usepackage{amsmath}
\usepackage{hyperref}

\usepackage{epsf}
\usepackage{epsfig}

\newtheorem{theorem}{Theorem}
\newtheorem{claim}{Claim}
\newtheorem{corollary}{Corollary}
\newtheorem{proposition}{Proposition}
\newtheorem{lemma}{Lemma}
\newtheorem{remark}{Remark}

\newcommand{\E}{\mathbb E}
\newcommand{\Q}{\mathbb Q}

\newcommand{\be}{\begin{equation}}
\newcommand{\ee}{\end{equation}}

\newenvironment{proof}[1][Proof]{\noindent\textbf{#1.} }{\ \rule{0.5em}{0.5em}}

\usepackage{graphicx}
\graphicspath{{converted_graphics/}}
\begin{document}

\title{On the RND under Heston's stochastic volatility model}
\author{Ben Boukai \\
Department of Mathematical Sciences, IUPUI \\
Indianapolis, IN 46202 , USA}
\maketitle

\begin{abstract}

\noindent We consider Heston's (1993) stochastic volatility model for valuation of  European options to which (semi) closed form solutions are available and are given in terms of characteristic functions. We prove that the class of scale-parameter distributions with mean being the forward spot price satisfies Heston's solution. Thus, we show that any member  of this class  could be used for the direct risk-neutral valuation of the option price under Heston's SV model. In fact, we also show that any RND with mean being the forward spot price that satisfies Hestons' option valuation solution, must be a member of a scale-family of distributions in that mean. As particular examples, we show that one-parameter versions of  the {\it Log-Normal, Inverse-Gaussian, Gamma, Weibull} and the {\it Inverse-Weibull} distributions  are all members of this class and thus provide explicit risk-neutral  densities (RND) for Heston's pricing model. We demonstrate, via exact calculations and Monte-Carlo simulations,  the applicability and suitability of these explicit RNDs using already published Index data with a calibrated Heston model ({\tt S\&P500}, Bakshi,  Cao and Chen (1997), and  {\tt ODAX,} Mrázek and Pospíšil (2017)), as well as current option market data ({\tt AMD}).

\bigskip

\textit{Keywords}: Heston model, option pricing, risk-neutral valuation, calibration.
\end{abstract}

\section{Introduction }

The stochastic volatility model for option valuation of Heston (1993) is widely accepted nowadays by both, academics and practitioners. It prescribes, under a risk-neutral probability measure $\Q$, say, the dynamics of the spot's (stock, index) price process $S=\{S_t, \, t\geq 0\}$, in relation to a corresponding, though unobservable (untradable ) volatility process $V=\{V_t, \, t\geq 0\}$  via a system of stochastic deferential equations. This system is given by
\be\label{1}
\begin{aligned}
dS_t= & rS_tdt +\sqrt{V_t}S_t dW_{1,t}\\
dV_t= & \kappa(\theta-V_t)+\eta\sqrt{V_t}dW_{2,t}, 
\end{aligned}
\ee
where $r$ is the risk-free interest rate,  $\kappa, \ \theta$ and $\eta$ are some constants (to be discussed below) and where  $W_{1}=\{W_{1,t}, \, t\geq 0\}$ and  $W_{2}=\{W_{2,t}, \, t\geq 0\}$ are two Brownian motion processes under $\Q$ with $d(W_{1}W_{2})=\rho dt$.  

The quest to incorporate a non-constant volatility in the option valuation model, has risen in the literature  (e.g., Wiggins (1987) or Stein and Stein (1991)) ever since the seminal work of Black and Scholes (1973)  and of Merton (1973), (abbreviated here as the BSM) in modeling the price of a European call option when the spot's price  was assumed to evolve, with a constant volatility of the spot's returns, $\sigma$,   as a geometric Brownian motion,
\be\label{2}
dS_t=  rS_tdt +\sigma S_t dW_{1,t}. 
\ee
Coupled with an ingenious argument of instantaneous portfolio hedging (along with other  assumptions such as self-financing, no-cost trading/carry, etc.) and an  application of Ito's Lemma to the underlying PDE, the BSM model provides an exact solution for the price of an European call option $C(\cdot)$.  Specifically, given the {\it current} spot price $S_\tau=S$ and the risk-free interest rate $r$, the price of the corresponding call option with price-strike $K$ and duration $T$, 
\be\label{3}
C_{S}(K)= S\, \Phi(d_1)-K\,  e^{-rt}\, \Phi(d_2),  
\ee
where $t=T-\tau$ is the {\it remaining} time to expiry. Here, using the conventional notation, $\Phi(\cdot)$ and $\phi(\cdot)$ denote the standard Normal $cdf$   and  $pdf$, respectively,  and
\be\label{4}
d_1:=\frac{-\log(\frac{K}{S})+(r+\frac{\sigma^2}{2})t}{\sigma\sqrt{t}} \qquad \text{and} \qquad d_2:=d_1(k)-\sigma\sqrt{t}.  
\ee

In similarity to the {\it form} of the  BSM solution in (\ref{3}), Heston (1993)  obtained  that the solution to the  system of PDE resulting from the stochastic volatility model, (\ref{1}), is given by
\be\label{5}
C_{S}(K)= S\,  P_1-K\,  e^{-rt}\,  P_2,  
\ee
where $P_j$ $j=1,2$, are two related (under a risk-neutral probability measure $\Q$)  conditional probabilities that the option will expire in-the-money, conditional on the given current stock price $S_\tau=S$ and the current volatility, $V_\tau=V_0$.  However, unlike the explicit BSM solution in (\ref{3}) which is given in terms of the normal (or log-normal) distribution, Heston (1993) provided (semi) closed-form solutions to these two probabilities, $P_1$ and $P_2$ in terms of their characteristic functions  (for more details, see the Appendix). Hence, $C_{S}(K)$ in (\ref{5}) is readily computable, via complex integration,  for any choice of the parameters $\vartheta=(\kappa, \theta, \eta, \rho)$ in (\ref{1}), all in addition to $S, \ V_0$ and $r$.  These parameters have particular meaning in context of the SV  model (\ref{1}): $r$ is the prevailing risk-free interest rate; $\rho$ is the correlation between the two Brownian motions comprising it; $\theta$ is the long-run average, $\kappa$ is the mean-reversion speed and $\eta^2$ is the variance of the volatility $V$ (see also Section 4 below). It should be noted that different choices of $\vartheta$ will lead to different {\it values} $C_{S}(K)$ in (\ref{5}) and hence, the value $\vartheta=(\kappa, \theta, \eta, \rho)$ must be appropriately {\it `calibrated'} first for $C_{S}(K)$ to actually match the option market data.

The role of the risk-neutral probability measure $\Q$ in option valuation in general and in determining the specific solution in (\ref{5})  (or in (\ref{3})) in particular, cannot be overstated (in the `risk-neutral' world).  As was established by Cox and Ross (1976), the risk-neutral equilibrium requires that for $T>\tau$ (with $t=T-\tau$), 
\be\label{6}
\begin{aligned}
\E (S_T| S_\tau) =  &  \int S_T\,  d\Q(S_T)\\
= &  \int S_T\cdot q(S_T)d S_T= S_\tau e^{rt}
\end{aligned}
\ee
and that (in the case of a European call option), $C_{S_\tau}(K)$ must alo satisfy 
\be\label{7}
\begin{aligned}
C_{S_\tau}(K)= & e^{-rt}\E \left( \max(S_T-K, 0)\, \, | S_\tau \right) \\
= & e^{-rt} \int (S_T-K)^{+}\,  d\Q(S_T)\\
= & e^{-rt} \int_K^\infty (S_T-K)\,   q(S_T)dS_T,
\end{aligned}
\ee
where,  for any $x\in {\Bbb R}$,  $x^+:=max(x, 0)$. Here $q(\cdot)$ is the risk-neutral density (RND) under $\Q$, reflective of the conditional distribution of the spot price $S_T$ at time $T$, given the spot price, $S_\tau$ at time $\tau<T$.  The risk-neutral probability $\Q$ links together the option evaluation and the distribution of spot price $S_T$ and its stochastic dynamics governing the model (as in (\ref{1}), say). As was mentioned earlier, in the case of the BSM in (\ref{2}) the RND is unique and is given by the log-normal distribution. However, since the Heston (1993) model involves the dynamics of two stochastic processes, one of which (the volatility, V) is untradable and hence not directly observable, there are innumerable many choices of RNDs,  $q(\cdot)$, that would satisfy (\ref{6})-(\ref{7}) and hence, the general solutions of $P_1$ an $P_2$ in (\ref{5}) by means of characteristic functions (per each choice of $\vartheta=( \kappa, \theta, \eta, \rho)$).

Needless to say, there is an extensive body of literature dealing with the {\it `extraction', `recovery', `estimation'} or {\it  `approximation'}, in parametric or non-parametric frameworks,  of the RND, $q(\cdot)$ from the available (market) option prices; see  Jackwerth (2004), Figlewski (2010), Grith and Krätschmer  (2012)  and  Figlewski (2018),  for comprehensive reviews of the subject. With the parametric approach in particular, one strives to estimate by various means (maximum likelihood, method of moments, least squares, etc.)  the parameters of some {\it assumed}  distribution so as to approximate  available option data or implied volatilities (c.f. Jackwerth and Rubinstein (1996)). This type of {\it assumed} multi-parameter distributions includes some mixtures of log-normal distribution (Mizrach (2010), Grith and Krätschmer (2012)),  generalized gamma (Grith and Krätschmer  (2012)), generalized extreme value (Figlewski (2010)), the gamma and the Weibull distributions (Savickas (2005)), among  others. While empirical considerations have often led to suggesting these parametric distributions as possible RNDs, the motivation to these considerations did not include direct link to the governing pricing model and it dynamics, as was the case in the  BSM model, linking directly the log-normal distribution and the price dynamics of model (\ref{2}).

In this paper, we present in our Theorem \ref{T2} a more direct approach  (and hence,  a link)  to the RND quest in the case of Heston's (1993)  SV model (\ref{1}).   By expanding the last term in (\ref{7}), with $S_\tau=S$, we obtain, 
\be\label{7a}
C_{S}(K)=  e^{-rt} \int_K^\infty S_T\cdot q(S_T)dS_T-Ke^{-rt}\Q(S_T>K). 
\ee
Clearly, by comparing (\ref{5}) to (\ref{7a}), it follows that $P_2= \Q(S_T>K)\equiv 1-Q(K)$ (the risk-neutral probability of the option expiring in the money), whereas by, (\ref{6}) and (\ref{5}),
\be\label{7b}
P_1\equiv \int_K^\infty\frac{S_T}{Se^{rt}}\cdot q(S_T)dS_T=  \int_K^\infty\frac{S_T}{\E (S_T| S_\tau)}\cdot q(S_T)dS_T. 
\ee
We note that since by (\ref{6}),  we have, 
$$
\int_0^\infty\frac{S_T}{\E (S_T| S_\tau)}\cdot q(S_T)dS_T=1,
$$
the  probability $P_1$ is also being interpreted  (see for example xxxx)  as the probability of the option expiring in the money, but under the so-called physical probability measure that is being dominated by $\Q$. However here, in the case of of the Heston's (1993) model, we consider a different interpretation of this term $P_1$, which enables us to characterize a class of RND candidates that   satisfy (\ref{5}).

It is a standard notation to denote by  $\Delta(K)$  the so-call {\it delta} function (or hedging fraction) in the option valuation, as defined by
\be\label{9}
\Delta(K)=\frac{\partial C_{S}(K)}{\partial S}. 
\ee
In the Appendix, we show that for  Heston's call option price $C_S(K)$ as given in (\ref{5}), one has (see also  Bakshi,  Cao and Chen (1997)), 
\be\label{9a}
P_1\equiv \Delta(K). 
\ee
  Hence,  under model (\ref{1}), Heston's solution for the option price in (\ref{5}) can be written in an equivalent form as:
\be\label{8}
C_{S}(K) \equiv  S\cdot \Delta(K)-K\,  e^{-rt}\, \cdot (1-Q(K)).
\ee
We point out in passing that this presentation (\ref{8}) also trivially applies to the BSM option price in (\ref{3}) since in that case,  $\Phi(d_1)\equiv \Delta(K)$ and $\Phi(d_2)$ is just the probability that option will expire in the money (calculated under the log-normal distribution).  

In Section 2, we identify the class of distributions (and hence of RNDs) that admit the presentation in (\ref{8}) of Heston's (1993) option price as as is given in (\ref{5}). Specifically,  we show that any risk-neutral probability distribution  $Q$ that satisfies  (\ref{6})-(\ref{7}) with a {\it scale} parameter $\mu=S \cdot e^{rt}$ would admit the presentation in (\ref{8}) and hence would satisfy Hestons' (1993) option pricing model in (\ref{5}) . In fact, we also show in the Appendix that the RNDs that may be calculated directly from Heston's characteristic functions (corresponding to $P_1$ and $P_2$) are members of this class of distributions as well. In Theorem \ref{T2} below we establish the direct link, through Heston's (1993) solution in (\ref{5}) (or  (\ref{8})) between this class of RNDs and the assumed stochastic volatility model in (\ref{1}) governing the spot price dynamics. 

In Section 3 we provide some specific examples of well known distributions that satisfy (\ref{8}). These include the {\it Gamma, Inverse Gaussian,  Log-Normal}, the {\it Weibull} and the {\it Inverse Weibull} distributions (all with a particular parametrization) as possible RND solutions for  option valuation under Heston's  SV model. The extent agreement between each of these five particular distributions as a possible RND for the Heston's model, and the actual Heston's RND (calculated from $P_2$ in the Appendix) and the simulated distribution of the spot prices obtained under (a discretized version of) model (\ref{1}) is illustrated numerically in Section 4.  We demonstrated  the applicability and suitability of these explicit RNDs using already published Index data with a calibrated Heston model ({\tt S\&P500}, Bakshi,  Cao and Chen (1997), and  {\tt ODAX,} Mrázek and Pospíšil (2017)), as well as current option market data ({\tt AMD}).   In the Appendix, we present the expressions for for Heston's characteristic functions and discuss some of the immediate properties leading to our main results as stated in Theorem \ref{T2}.

\section{The scale parameter class of Heston's RND}
   
In this section we identify the class of distributions, and therefore of possible RNDs that admit the presentation in (\ref{8}) for the price of a European call option.  Specifically,  we show that any RND candidate that satisfies (\ref{6})-(\ref{7}) with a {\it scale} parameter $\mu=S \cdot e^{rt}$ would admit the presentation in (\ref{8}) and hence in light of the result in (\ref{9a})  (see Claim 7 in the Appendix), would equivalently satisfy Hestons' (1993) option pricing model in (\ref{5}).

\medskip

To that end and to simplify the presentation,   we consider a continuous positive random variable  $X$ with mean $\mu>0$ (with respect to some underlying probability measure $\Q$).   We denote by $Q_\mu(\cdot)$ and $q_\mu(\cdot)$ the $cdf$ and $pdf$ of $X$, respectively, to emphasize their dependency on $\mu$, as a parameter.  Similarly, for a given $\mu>0$, we write $E_\mu(\cdot)$ for the expectation of $X$ (or functions thereof) under $Q_\mu$ so that,
$$
\mu: =  E_\mu(X)= \int_{0}^\infty xq_\mu(x)dx \equiv  \int_{0}^\infty (1-Q_\mu(x))dx. 
$$
In similarity to (\ref{7}), we define, for each $s\geq 0$,  
$$c_\mu(s):= E_\mu[ (X-s)^+].$$
Clearly, $c_\mu(0)=\mu$.  Note that $c_\mu(\cdot)$ is merely the {\it undiscounted} version of $C_S(\cdot)$ in (\ref{7}), so that with $\mu=S \cdot e^{rt}$ as in (\ref{6}), we have, $c_\mu(K)\equiv e^{rt}\cdot C_S(K)$. 

It is straightforward to see that, as in (\ref{7a}), 
\be\label{11}
c_\mu(s)=\int_s^\infty (x-s)q_\mu(x)\, dx =\int_s^\infty xq_\mu(x)\, d\,x- s(1-Q_\mu(s)), 
\ee
or equivalently, 
\be\label{12a}
c_\mu(s)\equiv  \int_s^\infty (1-Q_\mu(x))\, dx.  
\ee
Hence, it follows immediately from (\ref{12a}) that for each $s>0$, 
\be\label{13}
c^\prime_\mu(s):= \frac{\partial c_\mu(s)}{\partial s}=-(1-Q_\mu(s)).
\ee
As we proceed to explore more of the basic properties of the function $c_\mu(\cdot)$, we add the simple assumption that $\mu$ is a scale-parameter of the underlying distribution of $X$. 
  \medskip
 
 \noindent \underbar{{\bf {Assumption A:}}} {\it We assume that ${\cal{Q}}:=\{Q_\mu,\  \mu>0\}$ is a {\it scale-family} of  distributions (under $\Q$), so that for any given 
 $\mu>0$
 $$
 Q_\mu(x)\equiv Q_1(x/\mu) \qquad \text{and}\qquad q_\mu(x)\equiv \frac{1}{\mu}q_1(x/\mu), \ \forall x>0,
 $$
 for some $cdf$ $Q_1(\cdot)$ with a $pdf$ $q_1(\cdot)$ satisfying   $\int_0^\infty x q_1(x)dx=1$ and \\ $ \int_0^\infty x^2 q_1(x)dx<\infty$.}

\medskip

In Lemma  \ref{L0} below we establish the linear homogeneity of  $c_\mu(s)$ under {\it Assumption A} and  provide  in Lemma \ref{L1} the implied re-scaling property  of this function   and the consequential specific  derived  form of $c_\mu(s)$ as presented in Theorem \ref{T1} below.   For the linear homogeneity property of the European options, in general,  see Theorems 6 \& 9 of Merton (1973) or Theorem 2.3 in Jiang (2005).

\begin{lemma} \label{L0} Suppose that  $Q_\mu\in {\cal Q}$ and thus it satisfies the conditions of  {\it Assumption A}, then the function $c_\mu(s)$ as defined in (\ref{12a}) (or  (\ref{11}) is homogeneous of degree one in $s$ and in $\mu$. That is, for $s^\prime = \alpha\, s$ and $\mu^\prime=\alpha\,  \mu$ with $\alpha >0$, we have $c_{\mu^\prime}(s^\prime) \equiv \alpha\, c_\mu(s)$. 
\end{lemma}
\begin{proof} With a simple change of variable, it follows  immediately from {\it Assumption A} and (\ref{12a}) that with $s^\prime = \alpha\, s$, $\mu^\prime=\alpha \, \mu$, $\alpha >0$, we have
$$
c_{\mu^\prime}(s^\prime)=  \int_{s^\prime}^\infty (1-Q_{\mu^\prime}(x))\, d\,x\equiv  \alpha\, \int_{s^\prime/\alpha}^\infty (1-Q_\mu(u))\, d\,u=\alpha \, c_\mu(s). 
$$
\hfill \end{proof}
  
\noindent Now, by applying the results of Lemma \ref{L0} with $\alpha=1/\mu$, we immediately obtain the following useful result.  
\begin{lemma} \label{L1} Suppose that  $Q_\mu\in {\cal Q}$ and thus it satisfies the conditions of  {\it Assumption A},  then  it holds that 
\be\label{14}
c_\mu(s)=\mu\,  c_1(s/\mu),
\ee
for any $s>0$, where $c_1$ is  as defined in (\ref{12a}), but with respect to $Q_1$,  
\be\label{15}
c_1(s)=\int_{s}^\infty (1-Q_1(u))du\qquad \text{with} \qquad c^\prime_1(s)=-(1-Q_1(s)). 
\ee
\end{lemma}

It should  be clear from the above results that this function,  $c_\mu(s)$,  can be re-scaled or "standardized" so that $c_\mu(s)/\mu$ is independent of $\mu$. In particular, if $s=b\, \mu$ for some $b>0$, then again by (\ref{14}), $c_\mu(b\, \mu)= \mu\, c_1(b)$.
\medskip

Next,  as in (\ref{9}), we define the {\it 'Delta'}-function corresponding to the function $c_\mu(s)$ in (\ref{11}) or (\ref{12a}), as $\Delta_\mu(s):= {{\partial}c_\mu(s)/ {\partial \mu}}$. In the next theorem we show that under {\it Assumption A}, $\Delta_\mu(\cdot )$ may be expressed in terms of the truncated mean of $X$ and the consequential representation of $c_\mu(\cdot)$.

\begin{theorem}\label{T1}  Suppose that  $Q_\mu\in {\cal Q}$ and thus it satisfies the conditions of  {\it Assumption A},  then for each   $s>0$, 
\be\label{16}
\Delta_\mu(s)=\frac{1}{\mu} \int_{s}^\infty xq_\mu(x)dx.
\ee
Further, $\Delta_\mu(s)\equiv \Delta_1(s/\mu)$, where  $\Delta_1(s):=\int_s^\infty uq_1(u)du\leq 1$ for any $s>0$. Hence,  $c_\mu(s)$ in (\ref{11}) may be written as as
\be\label{17}
c_\mu(s)=\,  \mu  \Delta_\mu(s)- s (1-Q_\mu(s))
\ee
\end{theorem}

\begin{proof} To prove (\ref{16}),  note that by Lemma \ref{L1}, (\ref{15}) and (\ref{11}), 
\be\label{18}
\begin{aligned}
\Delta_\mu(s)=  & \frac{\partial}{\partial \mu}[\mu c_1(s/\mu)]=   c_1(s/\mu) -\frac{s}{\mu} c_1^\prime(s/\mu)= \\
  = &  \,  \int_{s/\mu}^\infty uq_1(u)du \equiv   \frac{1}{\mu} \int_{s}^\infty xq_\mu(x)dx.\\
\end{aligned}
\ee
The second part  follows immediately from the first part and {\it Assumption A} and noting that $\Delta_1(s)\leq \int_0^\infty uq_1(u)du= 1$. Finally,  since by (\ref{16}), $\int_s^\infty xq_\mu(x)\, d\,x= \mu \cdot \Delta_\mu(s)$, the main result in  (\ref{17}) follows directly from (\ref{11}). 
\hfill\end{proof}

\medskip
An  immediate conclusion of Theorem \ref{T1} is that if $Q_\mu$ is a member of the scale-family of distribution ${\cal Q}$ (by {\it Assumption A}) then the functions $c_\mu(s)$ can easily be evaluated by calculating first the values of  $c_1(\cdot)$ for the ratio $b=s/\mu$.  Specifically, 
\be\label{19}
c_\mu(s)=\mu \,  c_1(s/\mu)\equiv \mu \,  \Delta_1(s/\mu)-s\,  (1-Q_1(s/\mu)
\ee

 The results of the Theorem, either as given in  (\ref{17}) or in (\ref{19}), can be used directly for the risk-neutral  valuation of European call option with a strike $K$ and a current spot price $S_\tau\equiv S$, providing the expression for $C_S(K)$ as is given in (\ref{8}). That is, if $Q\in {\cal Q}$, then with $\mu=S\, e^{rt}$ applied to (\ref{19}), we have 
\be\label{20}
\begin{aligned}
C_S(K):= & e^{-rt}\E((X-K)^+|\, S)=e^{-rt}c_\mu(K) \\
 =\,   &S\,  \Delta_1(K/\mu)-{K \,  e^{-rt}}\,  (1-Q_1(K/\mu)).\\ 
\end{aligned}
\ee
We summarize these findings in the following Corollary.  
\begin{corollary}\label{C1} For any risk-neutral distribution $Q_\mu$ that satisfies, in addition to (\ref{6})-(\ref{7}), also the conditions of {\it Assumption A}, with $\mu=S\, e^{rt}$, so that $Q_\mu\in {\cal Q} $ (for some $Q_1(\cdot)$), we have as in (\ref{8}) that 
$$
C_S(K)=  S\,  \Delta_\mu(K)-{K \,  e^{-rt}}\, (1-Q_\mu(K)), 
$$
where $ \Delta_\mu(K):=\Delta_1(K/\mu)$ and $Q_\mu(K):= Q_1(K/\mu)$. Hence $Q_\mu$ also satisfies Heston's (1993) option pricing model (and solution) as given in (\ref{5}).  
\end{corollary}

\begin{remark} In the case in which the  risk-neutral evaluation of the option includes a dividend with a rate $q$, then $\E(S_T|\, S)=S\, e^{(r-q)t}$ in (\ref{6}), in which case, by applying $\mu=S \, e^{(r-q)t}$ to (\ref{19}) we obtain
$$
C_S(K)=e^{-rt}\, c_\mu(K)= {S \, e^{-qt}} \,  \Delta_\mu(K)-{K \, e^{-rt}}\,  (1-Q_\mu(K)).
$$
\end{remark}
It should be clear from  (\ref{20}) that since the probability distribution $Q_\mu$ is assumed here to be a member of a scale family ${\cal Q}$, its values depend on $K$ and $S$ only through the ratio $K/S$. 
In the Appendix, we assert in Proposition \ref{P1}  that any risk  neutral probability distribution $Q_\mu$ that satisfies the solution (\ref{5}) for Heston's option pricing model, must also be a member of a scale family of distributions, with a scale parameter  $\mu=S \, e^{rt}$ (or $\mu=S \, e^{(r-q)t}$, in the case of a dividend yielding spot).  This assertion follows directly from the specific form of Heston's RND established in the Appendix which is given in terms of characteristic function corresponding to the term $P_2$ (see  (\ref{30}) and the subsequent comments there). Hence combined, the statements of Corollary {\ref{C1} and Proposition \ref{P1}, can be summarized in the following theorem.
\begin{theorem}\label{T2}
Let  $Q_\mu(\cdot)$ be any risk-neutral distribution  that satisfies (\ref{6})-(\ref{7}) with a corresponding RND  $q_\mu(\cdot)$ and with $\mu=S \, e^{rt}$.   Then $Q_\mu(\cdot)$  satisfies Heston's option pricing solution in (\ref{5}) (and equivalently in (\ref{8})) if and only if  $Q_\mu(\cdot)$ is member of a scale-family of distributions with a scale parameter $\mu$.  
\end{theorem}

\section{Examples of explicit RNDs for the Heston Model}

In view  of Theorem \ref{T2}, the quest for finding appropriate RND for Heston' SV model for a particular parametrization of $\vartheta=( \kappa, \theta, \eta, \rho)$ must be focused only on those members of a scale-family of distributions with a scale parameter $\mu=S \, e^{rt}$.  Accordingly, we provide in this section, five particular examples of well-known distributions, that satisfy the conditions of {\it Assumption A} and hence admit, per Corollary \ref{C1},  the presentation   (\ref{8}) for the Heston's option pricing model in (\ref{5}). These well-known distributions, namely, the {\it Log-Normal}, the {\it Gamma}, the {\it Inverse Gaussian},  the {\it Weibull} and the {\it Inverse Weibull} distributions, are re-parametrized under {\it Assumption A} to a standardized, one-parameter version having mean $1$ and a second moment that depends on a single  parameter $\nu>0$ (in fact, we take $\nu\equiv \sigma\sqrt{t}$, for some $\sigma>0$).  Due to their relative simplicity (involving only one parameter), we view these distributions as {\it inexpensive} RNDs, easy to obtain, to  calculate and calibrate as compared to the alternatives approaches available in the literature.  We note that while the gamma and the Weibull distribution were considered by Savickas (2002, 2005) for deriving `alternative' option pricing formulas, the motivation for the parametrization there was  made without regard to the spot price dynamics (but rather for fitting kurtosis and skewness) and therefore are different. 

With these standardized distributions in hand and the corresponding explicit expressions for $Q_1(\cdot)$ as obtained under {\it Assumption A}, we utilize (\ref{19}) (with $\mu\equiv 1$) to first calculate  in each case  the expression 
for 
$$
c_1(s)= \Delta_1(s)-s(1-Q_1(s)),
$$
which is then used to obtain, with  $\mu>0$, the expression for the {\it undiscounted} option price as,
$$
c_\mu(s)=  \mu\times \left[\Delta_1({{s/\mu}})-\frac{s}{\mu}\times(1-Q_1(s/\mu))\right].
$$
Finally, as in (\ref{20}), the corresponding expression for the call option price is obtained as $C_S(K)= e^{-rt}c_\mu(K)$  (with $ \mu=Se^{rt}, \ s=K$ and $\nu=\sigma\sqrt{t}$).  We point out again, that each of these five distributions  would satisfy as RND,  Heston's (1993) general solution for the valuation of a European call option as is given in (\ref{5}).   We begin with the log-normal distribution which results with the classical Back-Scholes option pricing model (as given in (\ref{3})-(\ref{4})) .

\subsection{{The Log-Normal RND} } Suppose that the random variable $U$ has the 'standard'  (one-parameter) log-normal  distribution having mean $E(U)=1$ and variance $Var(U)=e^{\nu^2}-1$, for some $\nu>0$, so that $W=\log(U)\sim {\cal N}( -\nu^2/2,\  \nu^2)$. Accordingly, the $pdf$ of $U$ is given by;
$$
q_1(u)=\frac{1}{u\nu} \times \phi\left(\frac{\log(u)+\nu^2/2}{\nu}\right), \ \ \ \qquad u>0,  
$$
and its $cdf$ 
$$
Q_1(u)=Pr(U\leq u) =\int_0^u q_1(s)ds= \Phi\left(\frac{\log(u)+\nu^2/2}{\nu}\right), \ \ \forall u>0. 
$$
It is straightforward to verify that if $X\equiv \mu U$ for some $\mu>0$, then the $pdf$ of $X$ is the 'scaled' version of $q_1$, namely, $q_\mu(x)= \frac{1}{\mu}\, q_1(x/\mu)$, so this distribution satisfies {\it Assumption A}.

Next, we calculate the expression of $\Delta_1(s)$ which upon using the relation $U\equiv e^W$,  becomes
$$
\begin{aligned}
\Delta_1(s):=  & \int_{s}^\infty uq_1(u)dx=  \int_{\log(s)}^\infty e^w\,  \phi\left(\frac{w+\nu^2/2}{\nu}\right)\frac{dw}{\nu}\\
= & \int_{\log(s)}^\infty  \phi\left(\frac{w-\nu^2/2}{\nu}\right)\frac{dw}{\nu}= 1- \Phi\left(\frac{\log(s)-\nu^2/2}{\nu}\right). 
\end{aligned}
$$
Hence, for the 'standardized' model we have that 
$$
\begin{aligned}
c_1(s)=  & \Delta_1(s)-s(1-Q_1(s)) \equiv \\
 &  \left[1-\Phi\left(\frac{\log(s)-\nu^2/2}{\nu}\right)\right]- s \left[1-\Phi\left(\frac{\log(s)+\nu^2/2}{\nu}\right)\right]\\
\end{aligned}
$$
Accordingly, by Lemma \ref{L1}  and (\ref{19}), $c_\mu(s)\equiv \mu\times c_1(s/\mu)$ and we therefore immediately obtain the following expression for $c_\mu(s)$ as, 
$$
\begin{aligned}
c_\mu(s)= & \mu\times \left[\Delta_1(s/\mu)-\frac{s}{\mu}\times(1-Q_1(s/\mu))\right]\\
= & \mu \times \left[1-\Phi\left(\frac{\log(s/\mu)-\nu^2/2}{\nu}\right)\right]- s \times \left[1-\Phi\left(\frac{\log(s/\mu)+\nu^2/2}{\nu}\right)\right]\\
\equiv & \mu \times \Phi\left(\frac{\log(\mu/s)+\nu^2/2}{\nu}\right)- s \times \Phi\left(\frac{\log(\mu/s)-\nu^2/2}{\nu}\right),\\
\end{aligned}
$$
where the last equality utilized the symmetry of the normal distribution.  Finally, to calculate under the log-normal RND the price of a call option  at a strike $K$ when the current price of the spot is $S$, we utilize the above expression, $c_\mu(s)$,  with $\mu\equiv S\, e^{rt}$, $s\equiv K$ and $\nu\equiv  \sigma\sqrt{t}$ to obtain, $C_S(K)=e^{-rt}c_\mu(K)$, which matches exactly the Black-Scholes call option price as is given in (\ref{3})-(\ref{4}).

\subsection{The Gamma RND} We begin with some standard notations.  We write $W\sim {\cal G}(\alpha, \lambda)$ to indicate that the random variable $W$ has the gamma distribution with a scale parameter $\lambda>0$ and a shape parameter $\alpha>0$, in which case we write $g(\cdot; \alpha, \lambda)$ and $G(\cdot; \alpha, \lambda)$ for the corresponding $pdf$ and $cdf$ of $W$, respectively. Recall that $E(W)=\alpha/\lambda$ and $Var(W)=\alpha/\lambda^2$.  Additionally, we denote by $\Gamma(\alpha):=\int_{0}^\infty y^{\alpha-1} e^{-y}dy$ the gamma function whose incomplete version is $\Gamma(\xi, \ \alpha):=\int_{0}^\xi y^{\alpha-1} e^{-y}dy$, is defined for any $\xi>0$. 

Now suppose that a random variable $U$ has the 'standard' (one-parameter) Gamma  distribution having mean $E(U)=1$ and variance $Var(U)=\nu^2$, for some $\nu>0$,  so that $U\sim {\cal G}( a,\  a)$ where we substituted $a\equiv 1/\nu^2$.  
Accordingly,  the $pdf$ of $U$ is given by
$$
q_1(u):=  g(u; a, a)=\frac{{a} (au)^{a-1}e^{-au}} {\Gamma(a)},\ \ \ \qquad u>0.
$$
and its $cdf$,  by
$$
Q_1(u)= Pr(U\leq u):= G(u; a, a)= \frac{\Gamma(au, a)}{\Gamma(a)},  
$$
for any $u>0$. It is straightforward to verify that if $X\equiv \mu U$ for some $\mu>0$, then the $pdf$, $q_\mu(\cdot)$ of $X$ is the 'scaled' version of $q_1(\cdot)$, and that {\it Assumption A} holds in this case too.  Next, we calculate the expression for $\Delta_1(s)$,
$$
\begin{aligned}
\Delta_1(s)=  & \int_{s}^\infty uq_1(u)du=  \int_{s}^\infty u\,g(u; a, a)du \\
= & \int_{s}^\infty  \frac{a(au)^{a}e^{-au}} {a\Gamma(a)}du= 1- \frac{\Gamma(as, a+1)}{\Gamma(a+1)}\equiv 1-G(s;  a+1, a). 
\end{aligned}
$$
Accordingly, we obtain for the 'standardized' Gamma model that 
$$
\begin{aligned}
c_1(s)=  & \Delta_1(s)-s(1-Q_1(s)) \equiv \\
= &  \left[1- \frac{\Gamma(as, a+1)}{\Gamma(a+1)}\right]- s \left[1-\frac{\Gamma(as, a)}{\Gamma(a)}\right]\\
= & \left[1- G(s; a+1, a)\right]-s \left[1-G(s; a, a)\right]\\
\end{aligned}
$$

Again,  by Lemma \ref{L1}  and (\ref{19}), $c_\mu(s)\equiv \mu\times c_1(s/\mu)$ and we therefore immediately obtain the following expression for $c_\mu(s)$ in this case of the Gamma model as, 
\be\label{22}
\begin{aligned}
c_\mu(s)= & \mu\times \left[\Delta_1({{s/\mu}})-\frac{s}{\mu}\times(1-Q_1(s/\mu))\right]\\
= & \mu \times \left[1- G({s}/{\mu}; a+1, a)\right]- s \times \left[1-G({s}/{\mu}; a, a)\right]\\
\end{aligned}
\ee
Finally, to calculate under this Gamma RND  the price of a call option  at a strike $K$ when the current price of the spot is $S$, we will utilize (\ref{22}) with $\mu\equiv S\, e^{rt}$, $s\equiv K$ and $\nu\equiv  \sigma\sqrt{t}$ (so that $a\equiv 1/\sigma^2 t$)  to obtain, $C_S(K)=e^{-rt}c_\mu(K)$. 

\subsection{The Inverse Gaussian RND}  Using standard notation we write $W\sim {\cal IN}(\alpha, \lambda)$ to indicate that the random variable $W$ has the Inverse Gaussian distribution with mean 
 $E(W)=\alpha$ and $Var(W)=\alpha^3/\lambda$.  Now suppose that a random variable $U$ has the 'standard' (one-parameter) Inverse Gaussian distribution having mean $E(U)=1$ and variance $Var(U)=\nu^2$, for some $\nu>0$,  so that $U\sim {\cal IN}( 1,\  1/\nu^2)$.   Accordingly, the $pdf$ and $cdf$ of $U$ are given by;
$$
q_1(u)={{1}\over {\nu u^{3/2}}}\times \phi\left({{u-1}\over {\nu \sqrt{u}}}\right), \qquad  u>0,
$$
and
$$
Q_1(u)= \Phi\left({{u-1}\over {\nu \sqrt{u}}}\right)+e^{{{2}/{\nu^2}}}\times  \Phi\left(-{{u+1}\over {\nu \sqrt{u}}}\right)\ \qquad \forall u>0.
$$
Again, one can verify that if $X\equiv \mu U$ for some $\mu>0$, then the $pdf$, $q_\mu(\cdot)$ of $X$ is the 'scaled' version of $q_1(\cdot)$ above, so that {\it Assumption A} holds in this case too. 

In the case of this distribution, the values of of $\Delta_1(s)= \int_{s}^\infty uq_1(u)du$ must be evaluated numerically which,  when combined with the expression of $Q_1(s)$ given above,  provide the values of 
$$
c_\mu(s)=  \mu\times \left[\Delta_1({{s/\mu}})-\frac{s}{\mu}\times(1-Q_1(s/\mu))\right],
$$
for any $\mu>0$. Here again, the corresponding values of the call option $C_S(K)$ may be obtained,exactly along the same lines as in the previous examples, with $\mu\equiv S\, e^{rt}$, $s\equiv K$ and $\nu=\sigma\sqrt{t}$. 

\subsection{The Weibull RND} Using standard notation we write $W\sim {\cal W}(\xi,  \lambda)$ to indicate that the random variable $W$ has the
Weibull distribution with  $cdf$ and $pdf$ which are of the form, 
$$
F_W(w)=1- e^{-(w/\lambda)^\xi}, \ \ \ \text{and} \ \ \ f_W(w)=\frac{\xi}{\lambda}(\frac{w}{\lambda})^\xi e^{-(w/\lambda)^\xi}, \ \ \ w>0,
$$
respectively, where $\lambda>0$ is the scale parameter and $\xi>0$ is the shape parameter.  The mean and variance of $W$ are given by, 
$$
E(W)={\lambda}h_1(\xi)\qquad \text{and} \qquad Var(W)={\lambda}^2(h_2(\xi)-h_1(\xi)^2),
$$
where, $h_j(\xi):=\Gamma(1+j/\xi), \ j=1, \dots, 4.$   Now suppose that a random variable $U$ has the 'standard' (one-parameter) Weibull distribution having mean $E(U)=1$ and variance $Var(U)=\nu^2$, for some $\nu>0$. That is, for a given $\nu>0$, we let  $\xi^*\equiv \xi(\nu)$ be the (unique) solution of the equation 
\be\label{23}
\frac{h_2(\xi)}{h_1^2(\xi)}=1+\nu^2,
\ee
in which case, $h_j^*\equiv h_j(\xi^*), \ j=1, 2$, $\lambda^*\equiv 1/h_1^*$ and $U\sim {\cal W}( \xi^*,\  \lambda^*)$.   Accordingly, the $pdf$ and $cdf$ of $U$ are given by,
$$
Q_1(u)=1- e^{-(u/\lambda^*)^{\xi^*}}, \ \ \ \text{and} \ \ \ q_1(u)=\frac{\xi^*}{\lambda^*}(\frac{u}{\lambda^*})^{\xi^*} e^{-(u/\lambda^*)^{\xi^*}}, \ \ \ u>0,
$$
Again, it can be easily verified  that if $X\equiv \mu U$ for some $\mu>0$, then the $pdf$, $q_\mu(\cdot)$ of $X$ is the 'scaled' version of $q_1(\cdot)$ above, so that {\it Assumption A} holds in this case too.  For this RND, the values of of $\Delta_1(s)$ can be obtained in a closed form as
$$
\Delta_1(s)= \int_{s}^\infty uq_1(u)du =1-\frac{\Gamma((s/\lambda^*)^{\xi^*}; 1+1/\xi^*)}{\Gamma(1+1/\xi^*)},
$$
which, together with  the expression of $Q_1(\cdot)$ given above,  provide the values of 
$$
c_\mu(s)=  \mu\times \left[\Delta_1({{s/\mu}})-\frac{s}{\mu}\times(1-Q_1(s/\mu))\right],
$$
for any $\mu>0$. Here again, the corresponding values of the call option $C_S(K)$ may be obtained,exactly along the same lines as in the previous examples, with $\mu\equiv S\, e^{rt}$, $s\equiv K$ and $\nu=\sigma\sqrt{t}$.

\subsection{The Inverse Weibull RND} In similarilty to the above example, we write $W\sim {\cal IW}(\xi,  \alpha)$ to indicate that the random variable $W$ has the Inverse 
Weibull distribution (see for example,  de Gusmão at. el. (2009) )  with $cdf$ and $pdf$ which are of the form, 
\be\label{25}
F_W(w)=e^{-(\alpha/w)^\xi}, \ \ \ \text{and} \ \ \ f_W(w)=\frac{\xi}{\alpha}(\frac{\alpha}{w})^{\xi+1} e^{-(\alpha/w)^\xi}, \ \ \ w>0,
\ee
respectively, where $\alpha>0$ is the scale parameter and $\xi>2$ is the shape parameter.  In this case, the  mean and variance of $W$ are given by, 
$$
E(W)={\alpha}\tilde h_1(\xi)\qquad \text{and} \qquad Var(W)={\alpha}^2(\tilde h_2(\xi)-\tilde h_1^2(\xi)),
$$
where, $\tilde h_j(\xi)\equiv h_j(-\xi)= \Gamma(1-j/\xi), \ j=1, \dots, 4. $ Here too, we let $U$ have the 'standard' (one-parameter) Inverse Weibull distribution with  mean $E(U)=1$ and variance $Var(U)=\nu^2$, for some $\nu>0$. That is, for a given $\nu>0$, we let  $\xi^*\equiv \xi(\nu)$ be the (unique) solution of the equation 
\be\label{26}
\frac{\tilde h_2(\xi)}{\tilde h_1^2(\xi)}=1+\nu^2,
\ee
in which case, $U\sim {\cal IW}( \xi^*,\  \alpha^*)$ with $\alpha^*=1/\tilde h_1(\xi^*)$.   Accordingly, the $pdf, \ q_1(u)$, and $cdf, \ Q_1(u)$, of $U$ are as given in  (\ref{25}), but with $\xi^*$ and $\alpha^*$.  Hence, we may proceed exactly along the lines of the previous example to calculate $c_1(s)$, and $c_\mu(s)$ and hence, $C_S(K)= e^{-rt}c_\mu(K)$  (with $ \mu=Se^{rt}, \ s=K$ and $\nu=\sigma\sqrt{t}$).

\subsection{On  Skewness and Kurtosis} As can be seen, the distribution in each of these five examples satisfies the conditions of {\it Assumption A} and hence by Corollary \ref{C1}  could potential serve  as RND for Heston's SV model (\ref{1}). These distributions are defined by a single parameter, namely $\nu\equiv\sigma\sqrt{t}$, that affects their features, such as {\it kurtosis} and {\it skewness}, and hence their suitability as RND for particular scenarios of the SV model (\ref{1})-- as are determined by the structural model parameter  $\vartheta=(\kappa, \theta, \eta, \rho)$ (more on this point in the next section). However, for sake of completeness and for future reference we provide in Table 1 the expression for the   {\it kurtosis} and {\it skewness} for these five distributions. 

It is interesting to note that, in the relevant parametric domain, all but the Weibull example, have positive {\it Skewness} measure.  It can be numerically verified that in the Weibull case, $\gamma_s(\xi(\nu))$ changes it's sign and is negative  once $\nu<0.3083511$ and that $\gamma_k(\xi(\nu))<3$ for $0.2007844< \nu < 0.4698801$ and is a largely leptokurtic distribution for $\nu>0.4698801$. Hence, the Weibull distribution would be particularly  useful when the implied RND is negatively skewed, such as in the cases when the spot is an Index. More on this point in the next section.  

\begin{table}[h]
\begin{center}
\caption{The {\it skewness} and excess {\it kurtosis} measures of the RND Examples 3.1-3.5 as functions of the single parameter $\nu\equiv \sigma \sqrt{t}$.}

\begin{tabular}{ccccccccccccc}
\hline
 Distribution  &$E(U)$ 	& $Var(U)$ &  $Skew$	& & $Exc. Kurtosis$ \  \\ \hline
 \\
 ${\cal G}( 1/\nu^2, \ 1/\nu^2)$	 &  1 	&  $\nu^2$ & $2\nu$	& &	$6\nu^2$	 \  \\ 
 \\
 ${\cal IN}( 1,\  1/\nu^2)$	 & 1	& $\nu^2$ & $3\nu$	& & $15\nu^2$	 \  \\ 
 \\
  ${\cal W}(\xi,\   1/g_1(\xi))^{*}$	 & 1	& $\nu^2$ & $\gamma_s(\xi)$ & &	 $\gamma_k(\xi)-3$ 	\  \\ 
  \\
 ${\cal IW}(\xi,\   1/h_1(\xi))^{**}$ & 1	& $\nu^2$ & $\gamma_s(-\xi)$ & &	 $\gamma_k(-\xi)-3$ \  \\ 
 \\
  ${\cal LN}( -\nu^2/2,\  \nu^2)$ 	 & 1 	& $e^{\nu^2}-1$ & $(e^{\nu^2}+2) \sqrt{e^{\nu^2}-1}$	& &	$ e^{4\nu^2}+2 e^{3\nu^2}+3e^{2\nu^2} -6$	  \  \\ 
 \\ \hline
\end{tabular}
\end{center}
\vskip -10pt
\small{$^*$ Here $\xi\equiv \xi(\nu)$ solves equation (\ref{23}); \\ $^{**}$ Here $\xi\equiv \xi(\nu)$ solves equation (\ref{26}) and it is assumed that $\nu$ is such that $\xi(\nu)>4$ for these expressions to be valid. In particular, with  $h_j(\xi)=\Gamma(1+j/\xi), \ j=1, \dots, 4$, 
$$
\gamma_s(\xi)=\frac{h_3(\xi)-3h_2(\xi)h_1(\xi)+2h_1^3(\xi)}{\left[h_2(\xi)-h_1^2(\xi)\right]^{3/2}}
$$
and
$$
\gamma_k(\xi)=\frac{h_4(\xi)- 4 h_3(\xi)h_1(\xi) +6 h_2(\xi) h_1^2(\xi)- 3 h_1^4(\xi)} {\left[h_2(\xi)-h_1^2(\xi)\right]^{2}}.
$$
} 
\end{table}

\section{Comparisons of the Heston's RNDs}
Having introduced in the previous section several examples of distributions that serve as possible RND for Heston's (1993) option valuation (\ref{5}) under the stochastic volatility model (\ref{1}), we dedicate this section to illustration of their applicability and their relative comparison. In the Appendix, we provide the closed-form expressions for Heston's $P_1$ and $P_2$ as are given in terms of their characteristic functions (see Heston (1993)). These terms enable us to compute, for given $S_\tau=S$, $V_\tau=V_0$ and $r$, and for each choice of  $\vartheta=(\kappa, \theta, \eta, \rho)$, Heston's call price $C_S(K)$ as in (\ref{5}) as well as Heston's RND as derived from the  characteristic functions of $P_2$ (see Appendix for details). Features of this distribution, such as {\it Kurtosis} and {\it Skewness} as are largely determined by $\eta$ and $\rho$, respectively (see Bakshi,  Cao and Chen (1997) ), would serve as guide for matching a particular proposed RND from among our five examples (see also Table 1). For instance, in cases which admit an RND with a distinct  negative skew, the {\it Weibull} distribution could be considered, whereas, in those cases with a distinct  positive  skew, the {\it Inverse Weibull} or the other distributions  discussed in Section 3 could be considered.

Additionally, we may simulate observations on $(S_T, V_T)$ from a discretized version of Heston's stochastic volatility process (\ref{1}) to obtain the simulated rendition of the marginal distribution of $S_T$. In light of the scaling property of the RND, we present, whenever convenient, the results in terms of the rescaled spot priced, $S^*=S_T/\mu$, where $\mu=Se^{rt}$ (see Corollary 4).  In the simulations, we employed either the (reflective version of) Milstein's  (1975) discretization scheme or Alfonsi's (2010)  implicit discretization scheme all depending on whether the so-called Feller condition, $\zeta:=\kappa \theta/\eta^2>1$, holds or not (see Gatheral (2006) for a discussion).  We note that $\zeta$ is intimately related to the conditional distribution of $V_T$ implied by the SV model (\ref{1}), (see Proposition 2 of Andersen (2008) for details).  In all cases we also included a comparison of the Monte-Carlo distribution of $S^*$ to the actual Heston's RND as was numerically  calculated  using (\ref{30}) under the `calibrated' values of $\vartheta=(\kappa, \theta, \eta, \rho)$. 

\bigskip

\begin{figure}[h] 
  \centering
  \includegraphics[width=7.0in,height=4.5in,keepaspectratio]{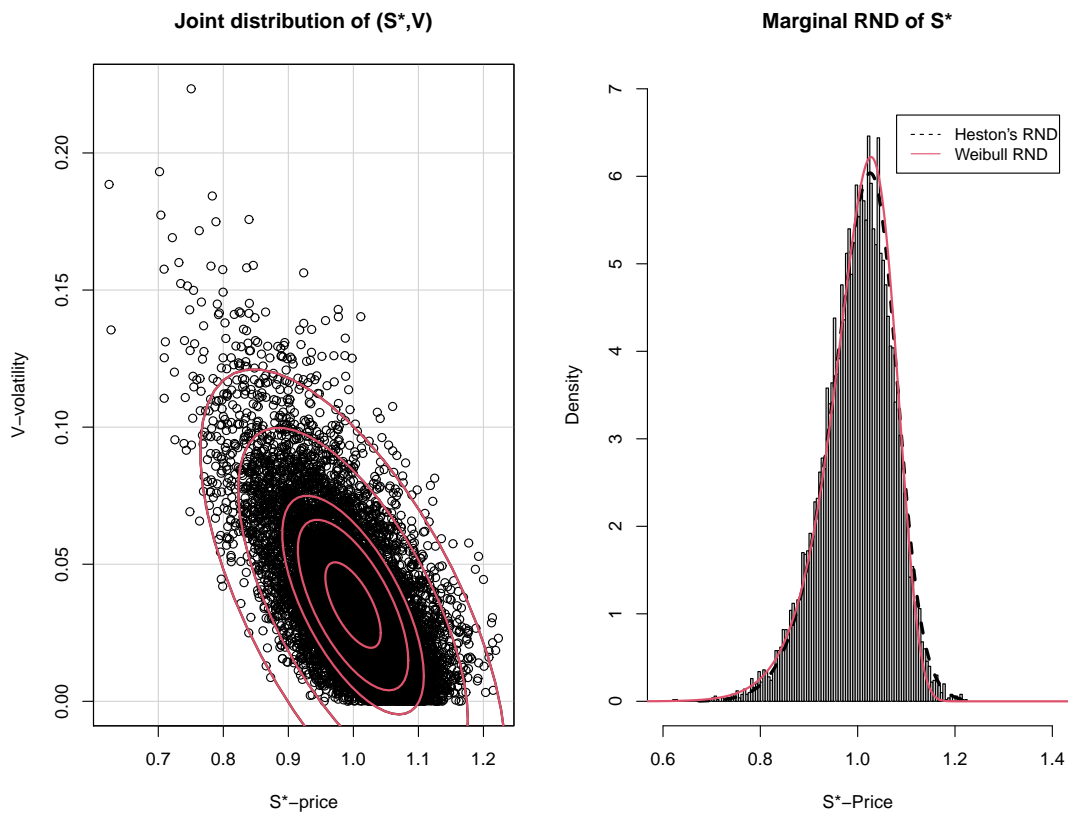}
  \caption[Figure 1]{\small{\it Simulated joint  $(S^*, V)$-distribution,  and the Heston's and Weibull RNDs for the {\bf S\&P 500} data based on the calibrated parameter  
  $\hat \vartheta =(1.15,\ 0.0347826,\  0.39,\  -0.64)$ as provided by  Bakshi,  Cao and Chen (1997). }}
  \label{fig:fig1}
\end{figure}

In the first two examples,  we use values of the structural parameters, $\vartheta=( \kappa, \theta, \eta, \rho)$, already calibrated to market data on  as can be found from Bakshi,  Cao and Chen (1997) (on the S\&P 500) and Mrázek and Pospíšil (2017) (on the ODAX).  
These two examples which involve market data on traded Indexes are  are used to illustrate the applicability of the Weibull distribution to situations in which the RND is negatively skewed and 
largely leptokurtic one.  Other similar illustrations using calibrated parameter values. such as from Lemaire ,  Montes and  Pagès (2020) ( on the E{\small URO} S{\small TOXX} 50), are also available but not presented here due to the limited space.   To allow for as realistic  as possible additional comparisons, the next example is based on current (as closing of December 31, 2020) market option data of  AMD. This  example serves to illustrate the applicability of the other RND candidates of Section 3, to situations exhibiting mild positive skewness. 

\bigskip
\noindent {\bf Example: S\&P 500:} Bakshi,  Cao and Chen (1997)   presented an extensive market data study for comparing several competing stochastic volatility models, including that of Heston's (1993). The data used covered options and spot prices for the S\&P 500 Index starting from June 1, 1988 through May 31, 1991.  From Table III there we find that  in addition to $r=0.02$, the `All Option'  estimated (or implied) structural parameter,  corresponding Heston's SV model, is
$$
\hat \vartheta =(1.15,\ (0.04/1.15),\  0.39,\  -0.64). 
$$
 In this case, $\hat \zeta= 0.526<1$, hence we used Milstein (reflective) scheme to obtain, for a short contract duration with $t=56/365=0.153$ year,  a Monte-Carlo sample of $M=10,000$ simulated pairs $(S^*,  V)$ with (standardized)  spot price and volatility, according to the SV model (\ref{1}).  Their joint distribution is presented in Figure 1a, where we have superimposed the matching 16\%, 50\%, 68\% ,  95\%  and 99.5\% contour lines. In Figure 1b we present the histogram of the simulated marginal distribution of the spot price $S^*$. The mean and standard deviation of these $M$ simulated spot price values are $\bar S^*=0.999462$ and $\hat \sigma\sqrt{t}=0.07213028$. We also included in the figure the curve of  the (implied, by $\hat \vartheta$)  Heston's RND as was computed directly by using (\ref{30}).   As is expected in the case of (risk-neutrality) modeling the spot prices of an Index,  the implied RND is {\bf negatively skewed} ($sk=-0.5018587$) which suggests a comparison against the Weibull distribution of Section 3.4. To that end, (and since we do not have the actual option data used by the authors) we simply matched $\nu$ to the `observed' value of $\hat \sigma\sqrt{t}$ and used it to obtain the numerical solution of equation (\ref{23}) as $\hat \xi= 17.40468$ at which point, $h_1(\hat \xi)=0.9699386$. For comparison, we also added to Figure 1b the plot of the ${\cal W}(\hat \xi, 1/h_1(\hat \xi))$ RND. As can be seen, the two RND  curves are almost indistinguishable. 
 
 \begin{figure}[h] 
  \centering
  \includegraphics[width=7.0in,height=4.5in,keepaspectratio]{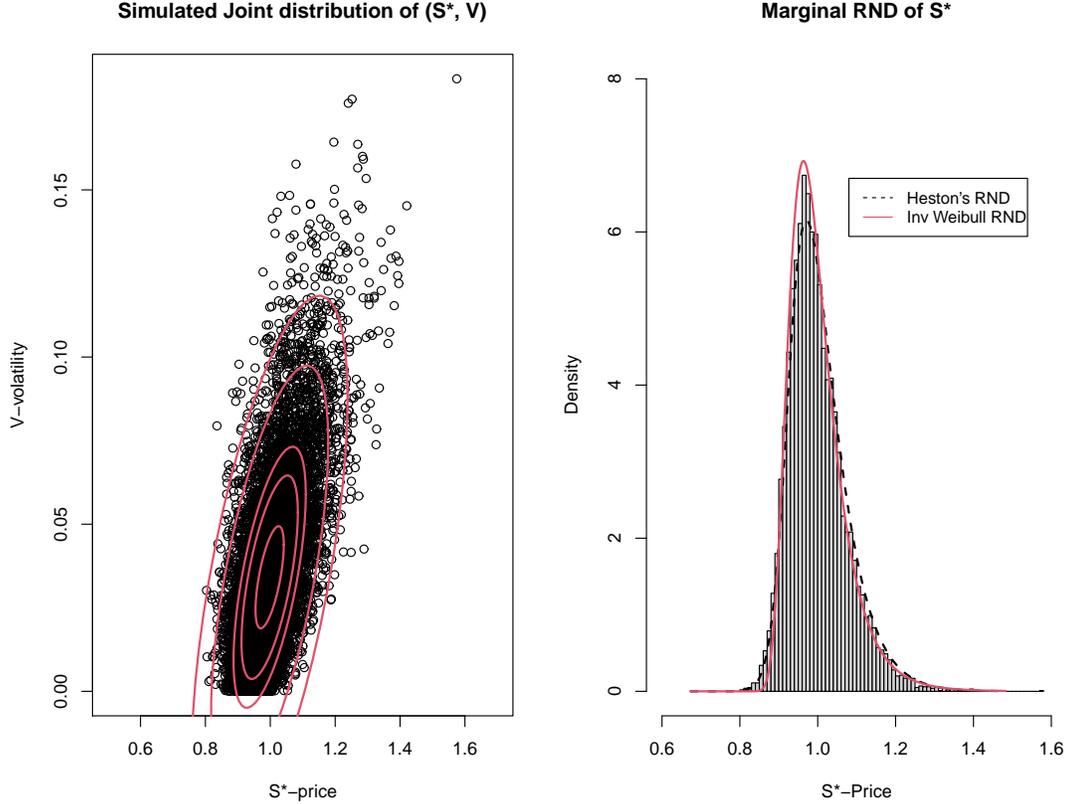}
  \caption[Figure 2]{\small{\it Illustrating the Inverse Weibull $pdf$  for the {\bf S\&P 500} data of Bakshi,  Cao and Chen (1997) calibrated parameter  {\bf but with hypothetically positive correlation} ($\rho=0.64$) resulting with a positively  skewed RND. }}
  \label{fig:fig1}
\end{figure}
 \bigskip
 To further illustrate the applicability of the proposed RND to situations with distinct {\bf positive}  {\it Skewness}, we considered again Bakshi,  Cao and Chen (1997)  calibrated parametrization used for Figure 1, but now with a ({\it hypothetically}) positive correlation, so that $\hat \vartheta =(1.15,\ (0.04/1.15),\  0.39,\  0.64).$  The simulated Monte-Carlo distributions (joint and marginal) are presented in Figure 2, exhibiting the distinct positive {\it Skewness} of  Heston's RND  (calculated from (\ref{30})).  This suggested a comparison to the Inverse Weibull distribution discussed in Section 3.5.   The mean and standard deviation of these $M$ simulated spot price values are $\bar S^*=0.9980681$ and $\hat \sigma\sqrt{t}=0.07379416$. Again, the value of $\nu$ was match to $\hat \sigma\sqrt{t}$ to obtain the solution of equation (\ref{26}) as  $\hat \xi= 18.16455$ at which point, $\tilde h_1(\hat \xi)=1.034936$. Accordingly, we added for comparison, the  plot of the ${\cal IW}(\hat \xi, 1/\tilde h_1(\hat \xi))$ RND to Figure 2b, illustrating the extent of the agreement between Heston's (implied) RND and the Inverse Weibull distribution in this case (with a distinct positives skew).

\begin{figure}[h] 
  \centering
  \includegraphics[width=7.0in,height=4.5in,keepaspectratio]{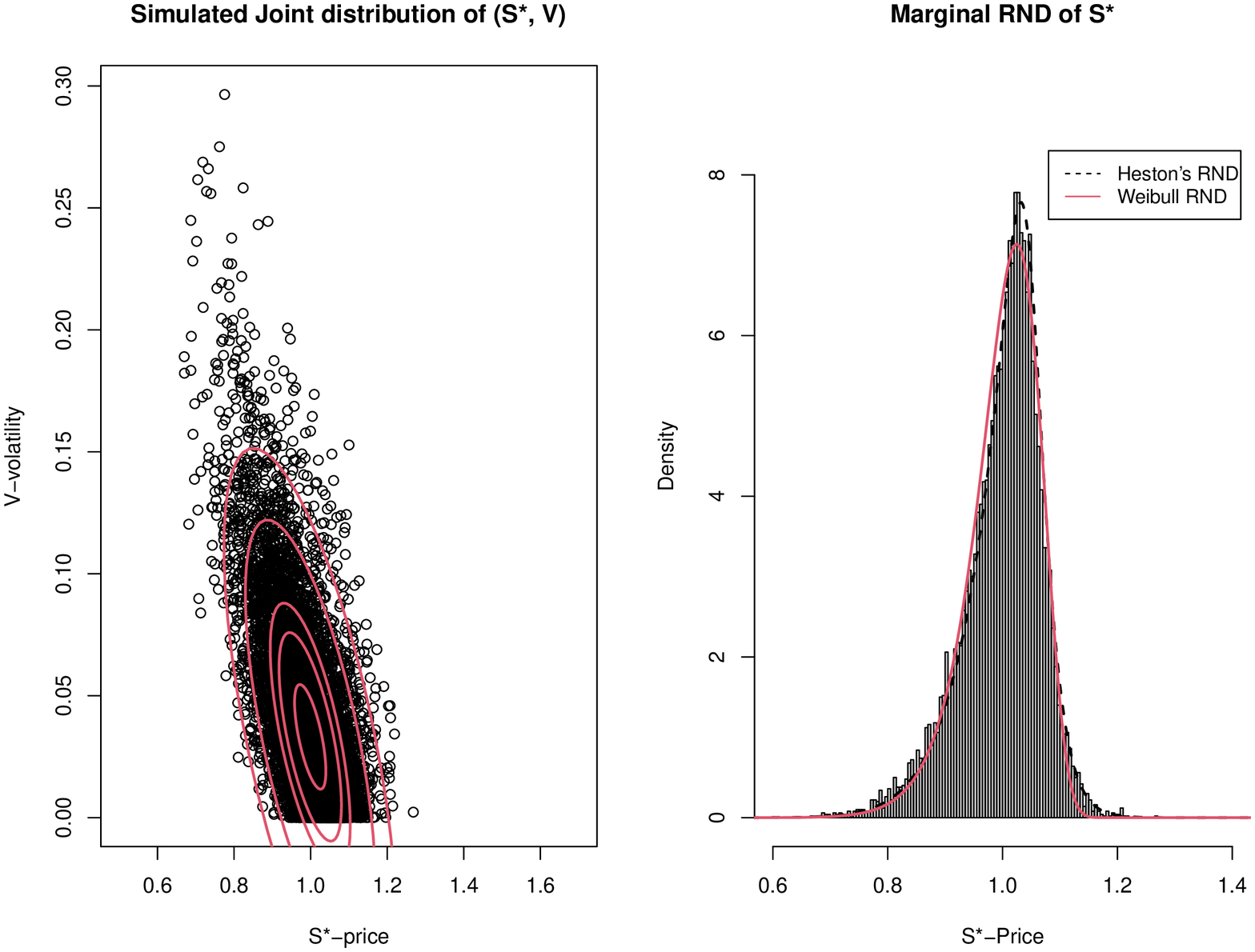}
  \caption[Figure 3]{\small{\it Simulated joint  $(S^*, V)$-distribution,  and the  (conditional) Heston's and Weibull RND for the {\bf ODAX} data based on the calibrated parameter  
  $\hat \vartheta =(1.22136,\ 0.06442,\  0.55993,\  -0.66255)$ as provided by Mrázek and Pospíšil (2017). }}
  \label{fig:fig1}
\end{figure}
 
 \bigskip
 
\noindent {\bf Example:  ODAX:} Mrázek and Pospíšil (2017)  studies various optimization techniques for calibrating and simulating the the Heston model. For demonstrate their results, they used the ODAX option Index with 5 blended maturities of three and six months and with 107 strikes,  as were recored on March 19, 2013.  The calibrated results of the structural parameter $\vartheta=(\kappa, \theta, \eta, \rho)$ are provided in  Table 4 there, 
 $$
\hat \vartheta =(1.22136, \ 0.06442, \  0.55993,\  -0.66255). 
$$
with $r=0.00207$ and with ``current'' $S=7962.31$ and with $V_0= 0.02497$.  Under this parametrization, we simulated with $t=64/365$, a total of $M=10,000$ pairs of $(S^*, \, V)$ to obtain from the discretized process, the renditions of their joint  distribution as well as the marginal distribution of $S^*$. These are presented in Figure 3. The mean and standard deviation of these $M$ simulated spot price values are $\bar S^*=0.9989159$ and $\hat \sigma\sqrt{t}=0.0692782$.  Superimposed, is the Heston's RND as was calculated, under this parametrization, from the characteristic function given the appendix. Again, as is expected for this index too,   the implied RND is {\bf negatively skewed} ($sk=-0.814462$) which suggests a comparison against the Weibull distribution. In this case too, the value of $\nu$ was match to $\hat \sigma\sqrt{t}$ to obtain the solution of equation (\ref{26}) as  $\hat \xi= 19.90341$ at which point, $\tilde h_1(\hat \xi)=0.9733867$. The graph of the  Weibull distribution, ${\cal W}(\hat \xi, 1/h_1(\hat \xi))$, is also displayed in the figure, indicating the excellent agreement, in this case two,  between the RNDs. 

\bigskip

\noindent {\bf Example:  AMD:} This example is based on real and current option data that we retrieved from {\it Yahoo Finance} as of the closing of trading on December 31, 2020.   The closing price of his stock, on that day, was $91.71$ and it pays no dividend, so that $q=0$ to add to the prevailing (risk-free) interest rate of $r=0.0016$. We chose this stock, AMD (Advanced Micro Devices Inc.) which is  a member of the technology sector, since it exhibits more directional risk to the upside, and hence with potentially, positively skewed RND.  

From the available option series, we selected the February 19, 2021 expiry, due to the relatively short contract with $t=47/365$ and some $N=39$ strikes, $K_1, \dots, K_{39}$  with corresponding call option (market) prices $C_1, \dots, C_{39}$ are available (we actually recorded the option prices as the average between the bid and ask). As standard measure of the {\it goodness-of-fit} between the model-calculated option price $C^{\tiny{Model}}(K_i)$ and the option market price $C_i$, we used the {\it Mean Squared Error}, MSE,
$$
MSE(Model)= \frac{1}{N}\sum_{i=1}^{N}(C^{{{Model}}}(K_i)-C_i)^2
$$
To calibrate the Heston SV model, we used the {\tt optim($\cdot$)} function of R, to minimize $MSE(Heston)$  over the  model's parameter , $\vartheta=(\kappa, \theta, \eta, \rho)$ with the initial 
values of $(2, 0.5, 0.6, 0)$ and with $V_0=0.25$. The results of the calibrated values are 
$$
\hat \vartheta =(1.38164142,\ 1.06637168,\ 1.72832698,\ 0.07768964).  
$$

\begin{figure}[h] 
  \centering
  \includegraphics[width=7.0in,height=4.5in,keepaspectratio]{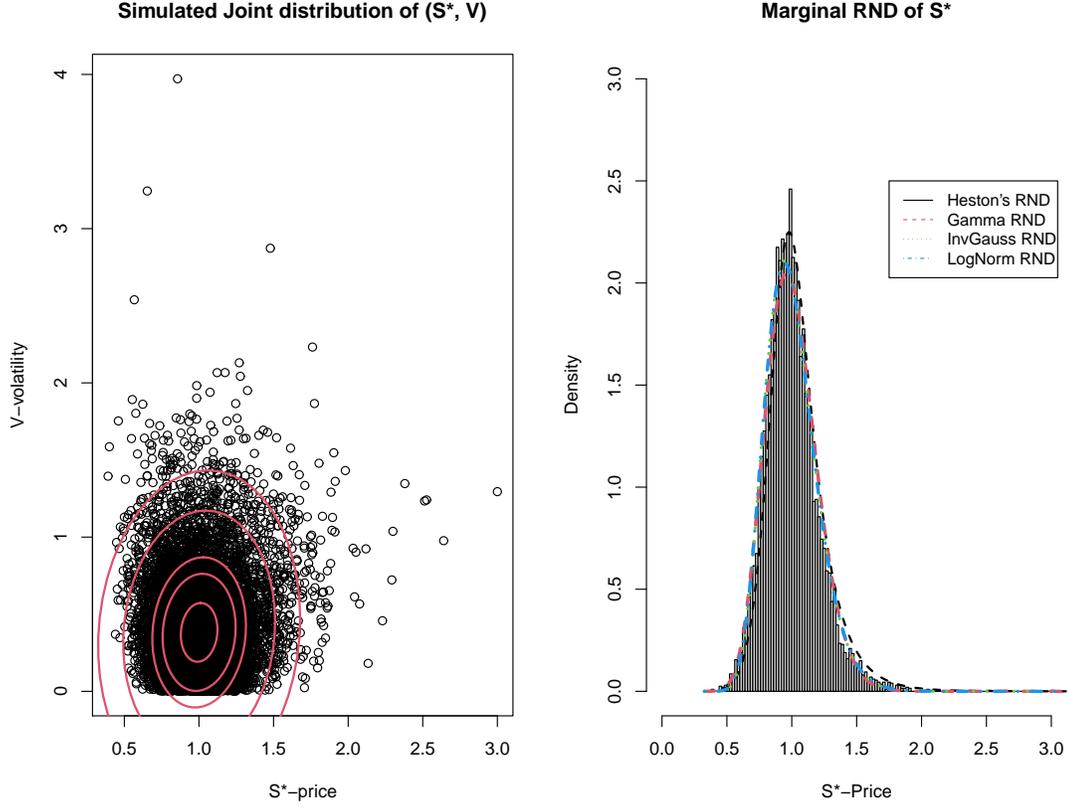}
  \caption[Figure 4]{\small{\it Simulated joint  $(S^*, V)$-distribution,  and Heston's, Gamma, Log-Normal and Inv. Gaussian RNDs  calculated based on the {\bf AMD} data of Table 2. }}
  \label{fig:fig1}
\end{figure}

This calibrated parameter, $\hat \vartheta $, was  then used to calculate, using Heston's characteristic function, the option prices according to Heston's SV model (\ref{5}).  These values are displayed  in Table 2, along with the actual market prices. Next, we obtained, as in the previous examples, a Monte-Carlo sample of $(S^*,\ V)$ whose results are displayed in Figure 4. The mean and standard deviation of these simulated stock prices are $\bar S^*=1.001237$ and $sd(S^*)=0.2079399$, respectively. As can be seen, the implied Heston's RND is, as was expected, {\bf positively skewed} ($sk=1.027201$). Accordingly, we considered those distribution from Table 1, as possible RND candidates in this situation.

\begin{figure}[h] 
  \centering
  \includegraphics[width=4.5in,height=3.5in,keepaspectratio]{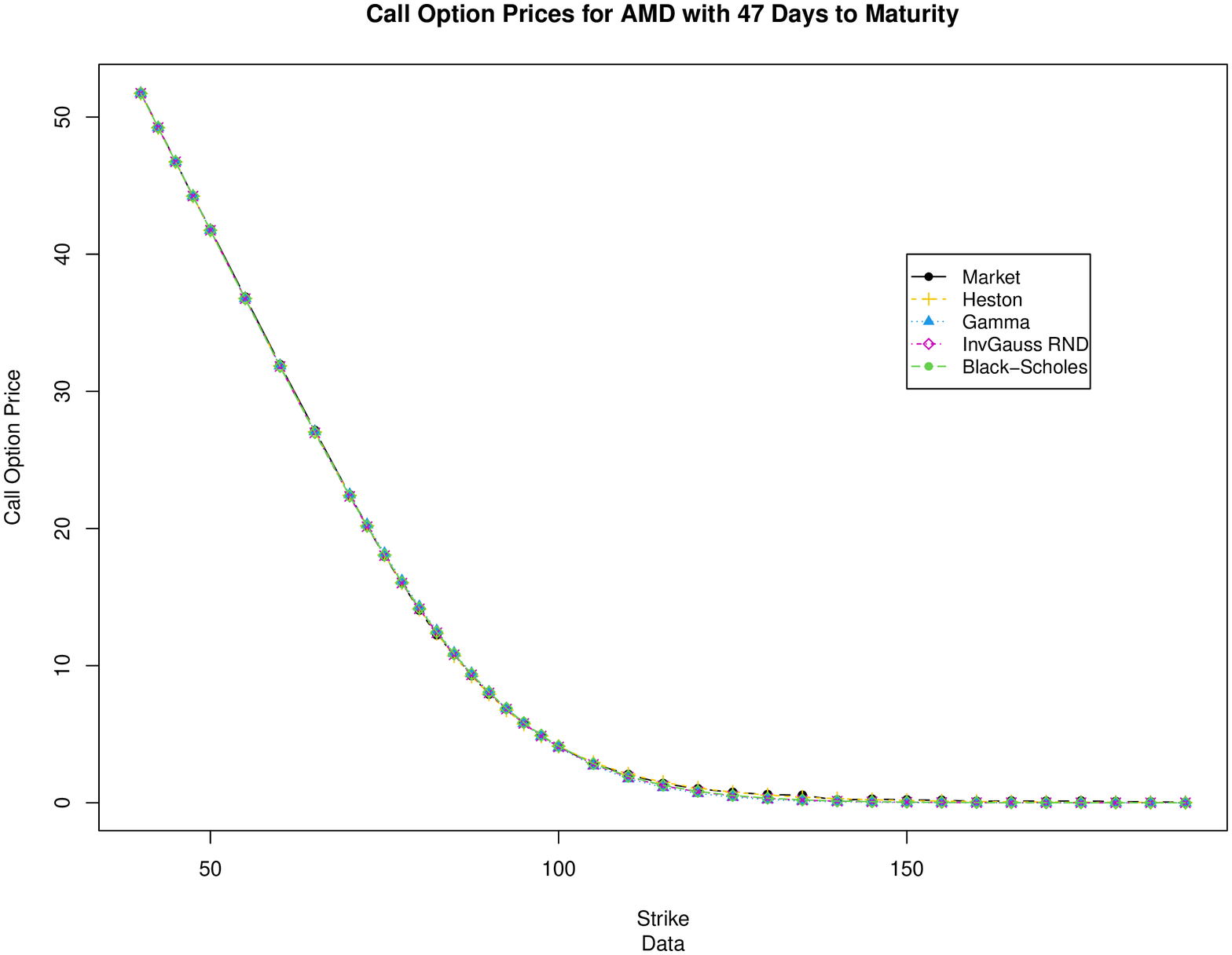}
  \caption[Figure 5]{\small{\it Comparing the the option prices obtained by the Heston, Gamma, Inverse Gaussian and the Black-Scholes Models for  the February 19, 2021 Option Series of {\tt  AMD} (with 47 DTE) as was quoted on the closing December 31, 2020 business day }}
  \label{fig:fig1}
\end{figure}

Since in this case, we have available the actual market option prices, we may estimate the parameter $\nu=\sigma\sqrt{t}$, defining these distribution directly, using  the `standard' Black-Scholes implied volatility, namely $\hat \nu=IV^{BS} \sqrt{t}$. This  entails using the the {\tt optim($\cdot$)} function again to minimize the $MSE(BS)$ with respect to the single parameter $\sigma$. This standard estimation procedure  yielded $IV^{BS}=0.550085$, so that $\hat\nu=0.1978301$. With this value at hand, we added to Figure 4b the graphs of the {\it Gamma}, {\it Inverse Gaussian} and the {\it Log-Normal} RND, as in Table 1. The extent of their agreement with the Heston's implied RND is self-evident. To further demonstrate that very point, we calculated the option prices under each one of these  modeled RND, and calculated  the corresponding  $MSE(Model)$. The results of this comparison are presented, side-by-side in Table 2.   In  Figure 5 we display the option price curve for each of these pricing models-- they are `virtually' almost identical in this example.

\begin{table}[t]
\begin{center}{\footnotesize
\caption{\small Comparing the the option prices obtained by the Heston, Gamma, Inverse Gaussian and the Black-Scholes Models for  the February 19, 2021 Option Series of {\tt  AMD} (with 47 DTE) as was quoted on the closing December 31, 2020 business day.}

\begin{tabular}{ccccccccccccc}
\hline
   & {\bf MSE} & {\bf 0.004410 } & {\bf 0.032725} &  {\bf 0.018126} & {\bf 0.016748 } \  \\ \hline

\hline
 Strike & Market Price & Heston &  Gamma & InvGaussan & Black.Scholes\\
\hline
40.0 & 51.775 & 51.720 & 51.738 & 51.737 & 51.737\\
\hline
42.5 & 49.275 & 49.222 & 49.239 & 49.238 & 49.238\\
\hline
45.0 & 46.775 & 46.726 & 46.741 & 46.739 & 46.739\\
\hline
47.5 & 44.200 & 44.231 & 44.244 & 44.240 & 44.240\\
\hline
50.0 & 41.825 & 41.741 & 41.751 & 41.742 & 41.743\\
\hline
55.0 & 36.875 & 36.779 & 36.783 & 36.758 & 36.762\\
\hline
60.0 & 31.950 & 31.869 & 31.871 & 31.816 & 31.824\\
\hline
65.0 & 27.150 & 27.058 & 27.073 & 26.977 & 26.993\\
\hline
70.0 & 22.450 & 22.423 & 22.476 & 22.339 & 22.365\\
\hline
72.5 & 20.200 & 20.203 & 20.285 & 20.132 & 20.164\\
\hline
75.0 & 17.975 & 18.070 & 18.184 & 18.022 & 18.059\\
\hline
77.5 & 16.025 & 16.039 & 16.186 & 16.022 & 16.064\\
\hline
80.0 & 14.050 & 14.127 & 14.302 & 14.145 & 14.190\\
\hline
82.5 & 12.250 & 12.349 & 12.544 & 12.399 & 12.449\\
\hline
85.0 & 10.800 & 10.719 & 10.917 & 10.793 & 10.846\\
\hline
87.5 & 9.275 & 9.242 & 9.428 & 9.330 & 9.385\\
\hline
90.0 & 7.925 & 7.923 & 8.077 & 8.010 & 8.066\\
\hline
92.5 & 6.850 & 6.760 & 6.866 & 6.830 & 6.888\\
\hline
95.0 & 5.800 & 5.746 & 5.790 & 5.786 & 5.843\\
\hline
97.5 & 4.925 & 4.871 & 4.844 & 4.870 & 4.927\\
\hline
100.0 & 4.100 & 4.120 & 4.021 & 4.073 & 4.129\\
\hline
105.0 & 2.835 & 2.939 & 2.706 & 2.799 & 2.852\\
\hline
110.0 & 2.065 & 2.096 & 1.766 & 1.880 & 1.928\\
\hline
115.0 & 1.410 & 1.498 & 1.119 & 1.237 & 1.278\\
\hline
120.0 & 1.025 & 1.075 & 0.688 & 0.798 & 0.833\\
\hline
125.0 & 0.765 & 0.776 & 0.412 & 0.506 & 0.535\\
\hline
130.0 & 0.605 & 0.563 & 0.240 & 0.315 & 0.338\\
\hline
135.0 & 0.550 & 0.411 & 0.136 & 0.194 & 0.211\\
\hline
140.0 & 0.205 & 0.302 & 0.076 & 0.117 & 0.131\\
\hline
145.0 & 0.265 & 0.224 & 0.041 & 0.070 & 0.080\\
\hline
150.0 & 0.215 & 0.167 & 0.022 & 0.042 & 0.049\\
\hline
155.0 & 0.185 & 0.125 & 0.011 & 0.024 & 0.029\\
\hline
160.0 & 0.110 & 0.094 & 0.006 & 0.014 & 0.017\\
\hline
165.0 & 0.135 & 0.072 & 0.003 & 0.008 & 0.010\\
\hline
170.0 & 0.120 & 0.055 & 0.001 & 0.005 & 0.006\\
\hline
175.0 & 0.135 & 0.042 & 0.001 & 0.003 & 0.004\\
\hline
180.0 & 0.095 & 0.032 & 0.000 & 0.001 & 0.002\\
\hline
185.0 & 0.070 & 0.025 & 0.000 & 0.001 & 0.001\\
\hline
190.0 & 0.040 & 0.020 & 0.000 & 0.000 & 0.001\\
\hline
\end{tabular}
}
\end{center}
\vskip -10pt
\small{\small Source: Yahoo Financial: \url{www.https://finance.yahoo.com/}
} 
\end{table}

\section{Appendix}
Heston (1993) provided (semi) closed form expressions to the probabilities $P_1$ and $P_2$ that comprise the solution $C_S(K)$ in (\ref{5}) to the option valuation under the stochastic volatility model (\ref{1}). Starting from a `guess' of the Black-Sholes style solution, 
\be\label{27}
C= SP_{1}-Ke^{-rt}P_{2},
\ee
he has shown that with $x:=\log{(S)}$,  this solution must satisfy the SDE resulting from the SV model in (\ref{1}), 
$$
\frac{\partial P_{j}}{\partial t}= \frac{1}{2}v\frac{\partial^{2}P_{j}}{\partial x^{2}}+\rho \eta v \frac{\partial^{2}P_{j}}{\partial x \partial v}+\frac{1}{2} \eta^{2} v \frac{\partial^{2} P_{j}}{\partial v^{2}}+ (r+u_{j}v)\frac{\partial P_{j}}{\partial x}+(a-b_{j}v)\frac{\partial P_{j}}{\partial v}, 
$$
for $j=1,2$, where $u_{1}=1/2 ,\ \ u_{2}=-1/2, \ \ b_{1}=\kappa - \rho \eta, \ \ b_{2}= \kappa $ and $a=\kappa\theta$. 
These closed form expressions are given by
\be\label{28a}
P_j=\frac{1}{2}+\frac{1}{\pi}\int_{0}^\infty{\cal R}\text{e}\left[ \frac{e^{-i\omega k}\psi_j(\omega, t, v, x)}{i\omega}\right]d\omega,
\ee
where  $k:=\log{(K)}$ and $\psi_j(\cdot)$ is the  characteristics function 
\be\label{28}
\psi_j(\omega, t, v, x):=\int_{-\infty}^\infty e^{i \omega s}p_j(s)ds\equiv e^{B_j(\omega, t)+D_j(\omega, t)v+i \omega x+ i\omega \,  rt},
\ee
where $p_j(\cdot)$ is the $pdf$ of $s_{T}=\log(S_T)$ corresponding to the probability $P_j, \ j=1, 2$ and 
$$
B_j(\omega, t)= \frac{\kappa \theta}{\eta^{2}}\{(b_{j} +d_j- i \omega \rho \eta)t -2\log(\frac{1- g_j e^{d_jt}}{1-g_j})\}
$$
$$
D_j(\omega, t)= \frac{b_{j} +d_j -i\omega \rho \eta}{\eta^{2}}(\frac{1-e^{d_jt}}{1-g_je^{d_jt}})
$$
$$
g_j=\frac{b_{j}-i\omega\, \rho \eta  +d_j}{b_{j}-i\omega\rho \eta  -d_j}
$$

$$
d_j= \sqrt{( i\omega \rho \eta - b_{j})^{2}- \eta^{2}(2 i\omega u_{j} -\omega^{2})}
$$
We point out that $d_j$ above is taken to be the positive root of the Riccati equation involving $D_j$.  However using instead the negative root, namely $d_j^{\prime}=-d_j$, was shown to provide an equivalent, but yet more stable solution for $\psi_j$ above -see Albrecher, Mayer,  Schoutens,   and Tistaer,  (2007) for more details on this so-called ``Heston Trap''.  In either case, efficient numerical routines such as the  
{\tt cfHeston} and {\tt callHestoncf} functions of the NMOF package of R, are readily available to accurately compute the values of $\psi_j$ and hence of $P_j$ and the call option values, for given $t, s$ and $v$ and any choice of $\vartheta=(\kappa, \theta, \eta, \rho)$.

 Now, having established (\ref{28}), the standard application of the Fourier transform provides (see for example Schmelzle (2010)) that the $pdf$ $p_j(\cdot)$ of $s_{T}=\log(S_T)$, can be obtained, for $s\in {\Bbb R}$, as
\be\label{29}
p_j(s)=\frac{1}{\pi}\int_{0}^\infty{\cal R}\text{e}\left[ {e^{-i\omega s}\psi_j(\omega, t, v, x)}\right]d\omega. 
\ee
Hence,  it follows immediately that the $pdf$ $\tilde p_j(\cdot)$ of $S_T$ is given by
$$
\begin{aligned}
& \tilde p_j(u)=  \frac{1}{u}\times p_j(\log (u))\equiv \\
 & \frac{1}{\pi}\int_{0}^\infty{\cal R}\text{e}\left[  \frac{e^{-i\omega \log(u)}\psi_j(\omega, t, v, x)}{u}\right]d\omega, \qquad u>0.
\end{aligned} 
$$
Further, since the characteristic functions $\psi_j$ in (\ref{28}) are affine in $x+rt= \log(S)+rt\equiv \log(\mu)$, where as in Corollary \ref{C1}, $\mu=Se^{rt}$, we may rewrite $\tilde p_j(u)$ above as
\be\label{30}
 \tilde p_j(u)  = \frac{1}{\mu\pi}\int_{0}^\infty{\cal R}\text{e}\left[  \frac{e^{-i\omega \log(u/\mu)}\tilde \psi_j(\omega, t, v)}{u/\mu}\right]d\omega, 
\ee
where 
$$
\log(\tilde \psi_j(\omega, t, v)):= \log(\psi_j(\omega, t, v, x)-i \omega x- i\omega \,  rt.
$$
We point out that in light of (\ref{7a})  that $\tilde p_2(\cdot)$ in (\ref{30}) is the RND (under $\Q$) for the Heston's (1993) model and can similarly be easily evaluated numerically along-side of evaluating $P_2$.  Indeed we have, 
$$
P_2\equiv \int_{K}^\infty \tilde p_2(u)du =\Q(S_T>K).
$$
It should be noticed  from  expression  (\ref{30})  that any  RND, $\tilde p_2(\cdot)$ of the Heston Model, and the corresponding risk neutral distribution $Q_\mu(\cdot)$ of $S_T$, constitutes a scale-family of distributions in $\mu=Se^{rt}$, so that it satisfies the terms of {\it Assumption A}.  This  assertion is summarized in Proposition $\ref{P1}$ below. 

\begin{proposition}\label{P1} Let $q_\mu(\cdot)$ be any RND with a corresponding risk neutral distribution  $Q_\mu(\cdot)$ that satisfies Heston's solution in (\ref{5}),  with $\mu=S \, e^{rt}$ then $q_\mu(\cdot)$ is 
of  the form given in  (\ref{30}) and therefore $Q_\mu(\cdot)$ must be a member of a scale-family of distributions in $\mu$.  
\end{proposition}

The result stated in the next claim is known, but the details are instructive to proving (\ref{9a}). 

\begin{claim} Let $\Delta(K)=\partial C/\partial S$ as  in (\ref{9}), then for the Heston solution (\ref{5})  (or (\ref{27})) with $P_j$, $j=1, 2$  as are given in (\ref{28a}), we have $\Delta(K)=P_1$.
\end{claim}

\begin{proof} By (\ref{27}) we have
\be\label{32}
\frac{\partial C}{\partial S}= P_1+S\frac{\partial P_1}{\partial S}- Ke^{-rt}\frac{\partial P_2}{\partial S}. 
\ee
Since we have $x=\log(S)$, it follows from (\ref{28a}) that, 
$$
\frac{\partial P_j}{\partial S} = \frac{1}{\pi}\int_{0}^\infty{\cal R}\text{e}\left[ \frac{e^{-i\omega k}\psi_j(\omega, t, v, x)}{S}\right]d\omega. 
$$
Hence, by (\ref{29}) we have with $\mu=Se^{rt}$, 
$$
S\frac{\partial P_j}{\partial S} \equiv p_1(k), \qquad \text{and} \qquad Ke^{-rt}\frac{\partial P_2}{\partial S}\equiv \frac{K}{\mu}p_2(k), 
$$
which by (\ref{7b}) implies that $S{\partial P_1}/{\partial S}- Ke^{-rt}{\partial P_2}/{\partial S}=0$  in (\ref{32}). 
\end{proof}

\end{document}